\pdfoutput=1

\documentclass[letterpaper, 10 pt, conference]{ieeeconf}  

\usepackage{amsmath, amsfonts, amssymb, mathtools, commath, hyperref, cancel, bm, xcolor, microtype, etoolbox, comment,mdframed}
\DeclareMathOperator{\sgn}{\operatorname{sgn}}

\IEEEoverridecommandlockouts                              

\title{\LARGE \bf
A New Type of Nonlinear Disturbance Rejection
}

\author{Simon Kuang$^{1}$ and Xinfan Lin$^{1}$
\thanks{*This work was supported by the National Science Foundation CAREER Program (Grant No. 2046292)}
\thanks{$^{1}$ Department of Mechanical and Aerospace Engineering, University of California, Davis; Davis, CA 95616, USA. {\ttfamily \{slku, lxflin\}@ucdavis.edu}}}%

\newtheorem{proposition}{Proposition}
\newtheorem{definition}{Definition}
\newtheorem{theorem}{Theorem}
\newtheorem{remark}{Remark}
\newtheorem{corollary}{Corollary}
\newenvironment{after}{}{}

\DeclareMathOperator{\arsinh}{arsinh}

\newtoggle{preprint}
\settoggle{preprint}{true}
\iftoggle{preprint}{%
   \newenvironment{skippedproof}[1]{%
      \begin{proof}{#1}%
   }{%
      \end{proof}%
   }%
}{\excludecomment{skippedproof}}

\begin{document}

\maketitle
\thispagestyle{empty}
\pagestyle{empty}

\begin{abstract}
Asymptotic disturbance rejection (equivalently tracking) for nonlinear systems has been studied only in qualitative terms (the state is asymptotically stable under bounded disturbances).
We show how to prove quantitative performance guarantees for the nonlinear servomechanism problem.
Our technique originates by applying a gain inequalities point of view to an \emph{ad fontes} reexamination of the linear problem:
what is the nonlinear equivalent of a sensitivity transfer function with a zero at the origin?
We answer: a nonlinear input-output system is high-pass if its output is stable with respect to the \emph{derivative} of the input.
We first show that definition generalizes high-pass resistor-capacitor circuit analysis to accommodate nonlinear resistors.
We then show that this definition generalizes the steady-state disturbance rejection property of integral feedback controllers for linear systems.
The theoretical payoff is that low-frequency disturbance rejection is captured by a quantitative, non-asymptotic output cost bound.
Finally, we raise theoretical questions about compositionality of nonlinear operators.
\end{abstract}

\section{Introduction}
A \emph{high-pass} transfer function is one that vanishes at the origin.
This algebraic characterization applies to continuous-time, linear time-invariant input-output systems.
We know of no analog in the practically and theoretically relevant province of continuous-time nonlinear input-output systems.

This theoretical gap invites basic contemplation in rational mechanics, which can be intrinsically interesting.
Moreover, the analytical methods laid out in this paper are able to provide quantitative and physically meaningful answers to the nonlinear generalizations of the following linear problems:
\begin{itemize}
        \item \textbf{linear} What is the power consumption of a resistor-capacitor circuit under an AC voltage source?
        (Can be solved using transfer function methods.)

        \textbf{nonlinear} What is the power consumption of a diode-capacitor circuit under an AC voltage source?
        (Solved in \S\ref{sec:nonlinear_high_pass_example}.)

        \item \textbf{linear} Synthesize a closed-loop linear controller with steady-state disturbance rejection, and guarantee it in a given given system norm.
        (Can be solved using convex linear system parameterization techniques.)

        \textbf{nonlinear} Synthesize a closed-loop nonlinear controller with steady-state disturbance rejection, and guarantee it in a given system norm.
        (Solved in \S\ref{sec:robust_control}.)
\end{itemize}


\subsection{Related work}
\subsubsection{Systems theory}
Our characterization of nonlinear high-pass behavior combines an \(L^p\) form of Barbalat's lemma \cite{teel_asymptotic_1999}
with dissipativity theory, a set of Lyapunov-like techniques to establish input-output properties of nonlinear systems \cite{willems_dissipative_1972} \cite[Chapter 4]{brogliato_dissipative_2020}.
The \(L^p\) form of Barbalat's lemma appears to be less widely used than the \(L^\infty\) form useful in adaptive and time-varying control, which traditionally emphasize \(L^\infty\)-style results such as bounded-input, bounded-output or input-to-state stability.
A connection between dissipativity and stability, though not the one we draw, can be found in \cite{faulwasser_continuous-time_2021}.

Dissipativity has become the tool of choice for the question template ``what is a nonlinear X?'' where X stands for a linear entity such as an Ohmic resistor \cite{sepulchre_incremental_2022}.
In this paper, we set X equal to ``high-pass filter.''

\subsubsection{Robust control}
It is a classical objective in linear control that the output of a closed-loop control system should have zero sensitivity to steady-state disturbances.
The disturbance-to-output transfer function is one of the Gang of Four/Six analyzed in a feedback control system \cite[Chapter 12]{astrom_feedback_2008}.
When designing linear controllers by hand, adding an integral action term to the feedback is one way to ensure that the disturbance-to-output transfer function has a zero at the origin, which makes it high-pass.

In nonlinear control, just as in linear, robustness to constant disturbances can be achieved by adding an integral action term \cite[p.~478]{khalil_nonlinear_2002} \cite{khalil_universal_2000,mahmoud_asymptotic_1996}.
Techniques for achieving robustness to steady-state disturbances include \emph{forwarding} \cite{praly_stabilization_2001,kaliora_nonlinear_2004,astolfi_integral_2017,giaccagli_sufficient_2022,giaccagli_lmi_2023}, forward invariance \cite{isidori_steady-state_2008}, and mixed-sensitivity control synthesis using a \(1/s\)-weighted disturbance \cite{su_use_2002}.

Recent work has deepened the emphasis on bounding the disturbance-to-output map.
System Level Synthesis reparameterizes controllers in terms of the closed-loop disturbance-to-control and disturbance-to-output maps, and can be applied to nonlinear systems \cite{anderson_system_2019,wang_system-level_2019,yu_achieving_2020,conger_nonlinear_2022,furieri_neural_2022}.
Another analytical technique to bound system norms (smaller is better) is time-domain integral quadratic constraints (IQCs), which generalize quadratic mixed-sensitivity synthesis to certain nonlinear systems \cite{megretski_system_1997,seiler_stability_2015}.



\subsection{Contributions}
\subsubsection{Systems theory}
In \S\ref{sec:nonlinear_high_pass_filter}, we propose a two-part definition of nonlinear high-pass behavior: one property generalizes a system norm characterization of linear high-pass behavior (Def.~\ref{def:nonlinear_high_pass_filter}), and a second generalizes the asymptotic uniform continuity property (Def.~\ref{def:asymptotically_nonlinear_high_pass_filter}).
In \S\ref{sec:filters}, we argue inductively that this generalization makes sense; we apply it first to a linear high-pass circuit, then to three variations that replace the Ohmic resistor with a nonlinear one.

\S\ref{sec:higher_order_nhp} investigates compositional properties of nonlinear high-pass filters.

 
\subsubsection{Robust control}
In \S\ref{sec:robust_control},
we analyze a class of nonlinear PI controllers\footnote{This class includes certain controllers with anti-windup behavior \cite{tarbouriech_anti-windup_2009,simon_robust_2023}.} for a scalar linear system by verifying that the disturbance-to-output map is nonlinear high-pass.
We conclude that these controllers are able to achieve steady-state disturbance rejection \(y \to 0\).

This kind of conclusion, namely, of asymptotic (or exponential and/or incremental) stability is where the story usually ends when analyzing nonlinear steady-state disturbance rejection.
Our analysis, which generalizes mixed-sensitivity synthesis with a \(1/s\) disturbance weight, has a weaker hypothesis and a stronger conclusion.
Rather than assuming a constant disturbance, we assume a time-varying disturbance in which a higher derivative incurs a greater cost.
The conclusion bounds a time-domain cost integral of the system output in terms of the disturbance's initial value and slew rate.
The initial value ends up having no impact on the long-run behavior of the system.

\section{Preliminaries}
We present nonlinear input-output systems in state space form:
\begin{subequations}
        \label{eq:state_space}
        \begin{align}
                \dot x(t) &= f(x(t), u(t)) \\
                y(t) &= g(x(t), u(t))
        \end{align}
\end{subequations}
for \(t \in [0, \infty)\),
  where \(x : [0, \infty) \rightarrow \mathbb{R}^{\mathsf{d}_x}\),
        \(u : [0, \infty) \rightarrow \mathbb{R}\),
        and \(y : [0, \infty) \rightarrow \mathbb{R}\).
The dynamics \(f\) and output law \(g\) are assumed to be sufficiently smooth for the computations of interest.

For \(p \in [0, \infty)\), the \(p\)-norm of a signal \(v: [0, \infty) \rightarrow \mathbb{R}\) is defined as
\begin{align*}
        \left\|v\right\|_p^p &= \int_0^\infty \left\|v(t)\right\|^p \dif t.
\end{align*}

By \([n\ldots m]\) we denote the integers from \(n\) to \(m\), inclusive.


\section{Overview of dissipativity theory}
The following material is based on \cite[Chapter 4]{brogliato_dissipative_2020}.
\begin{definition}
        Given a nonlinear input-output system of the form \eqref{eq:state_space} and a function \(s: \mathbb{R} \times \mathbb{R} \rightarrow \mathbb{R}\), the system is said to be dissipative if there exists a nonnegative function \(V: \mathbb{R}^{\mathsf{d}_x} \rightarrow \mathbb{R}\) such that
        \begin{align*}
                -\int_0^T s(y(t), u(t)) \dif t &\le V(x(0)) - V(x(T))
        \end{align*}
        for all \(T \geq 0\).
\end{definition}

The function \(s\) is called the \emph{supply rate} of the nonlinear system, and the function \(V\) is called the \emph{storage function}.
For example, if \(s(y, u) = \gamma^2 u^2 - y^2\) for some \(\gamma > 0\), then we have, for all \(T \geq 0\),
\begin{align*}
        \int_0^T y(t)^2 \dif t
        &\leq \gamma^2 \int_0^T u(t)^2 \dif t + V(x(0)) - V(x(T)) \\
        &\leq \gamma^2 \int_0^T u(t)^2 \dif t + V(x(0)).
\end{align*}
This certifies that up to a constant depending on the initial state, this system has a finite \(L^2\) gain of \(\gamma\) from \(u\) to \(y\); in other words, it satisfies the following definition \cite[Definition~4.8]{brogliato_dissipative_2020}, which generalizes \(H^2\) stability from linear systems theory:

\begin{definition}[Weakly finite-gain stable]
        The system \eqref{eq:state_space} is \emph{weakly finite-gain stable} (WFGS) with gain \(\gamma\) if
        for some \(p > 0\), for all \(T \geq 0\),
        \begin{align*}
                \del{\int_0^T \left|y(t)\right|^p \dif t}^{1/p} 
                &\leq 
                \del{\int_0^T \left|u(t)\right|^p \dif t}^{1/p} 
                + k,
        \end{align*}
        where \(k\) may depend on \(x(0)\).
\end{definition}

In practice, we apply dissipativity theory by taking \(s\) as prescribed and constructing \(V(x)\) satisfying
\begin{align*}
        \dot V(x) = \dod{}{t} V(x(t))
        &\leq
        s(y(t), u(t)).
\end{align*}
If \(V\) is positive definite and \(s = 0\), then we have Lyapunov stability as a special case of dissipativity.

\section{Definition of a nonlinear high pass filter}
\label{sec:nonlinear_high_pass_filter}
We reason by generalizing from a high-pass LTI system \(\mathcal{L} y = G(s) \mathcal{L} u\).
The transfer function \(G(s)\) is said to be high-pass of order \(k\) if it is stable and has a \(k\)th order zero at the origin.
Equivalently, we have the factorization \(G(s) = s^k G_0(s)\) where \(G_0(s)\) is stable;
we can move \(k\) derivatives onto the input:
\begin{align*}
        \mathcal{L} y = G_0(s) \mathcal{L} \cbr{u^{(k)}}.
\end{align*}
In this case, \(y\) is WFGS with respect to \( u^{(k)}\).

We take this to be the defining characteristic of a nonlinear high pass filter.
\begin{after}
\begin{definition}[NHP(1)]
\label{def:nonlinear_high_pass_filter}
        The system \eqref{eq:state_space} is \emph{first-order nonlinear high pass} (NHP(1))
        if
        \begin{align*}
               \dod{}{t} V(x(t), u(t)) \leq \beta(|\dot u(t)|) -  \alpha(|y(t)|) ,
        \end{align*}
        where \(V(x, u)\) is a nonnegative storage function and \(\alpha\) and \(\beta\) are i) continuous, ii) zero only at zero, and iii) nondecreasing.
\end{definition}
\end{after}
\begin{after}
\begin{definition}[ANHP(1)]
\label{def:asymptotically_nonlinear_high_pass_filter}
        The system \eqref{eq:state_space} is \emph{asymptotically first-order nonlinear high pass} (ANHP(1)) if it is first-order nonlinear high pass 
        and
        \begin{align*}
               \dod{}{t} V_1(x(t), u(t)) \leq \beta_1(|\dot u(t)|) - \alpha_1(|\dot y(t)|) ,
        \end{align*}
        where \(V_1(x, u)\) is a nonnegative storage function and \(\beta_1\) is i) continuous, ii) zero only at zero, and iii) nondecreasing,
        and iv) \(\alpha_1(s) \geq \gamma_1^p s^p\) for some \(p \geq 1\).
\end{definition}
\end{after}
\begin{theorem}
        Suppose that the system \eqref{eq:state_space} is ANHP(1).
        Then if \(\int_0^T \beta(|\dot u(t)|)\dif t < \infty\) and \(\int_0^T \beta_1(|\dot u(t)|)\dif t < \infty\), then \(y(t) \to 0\).
\end{theorem}
\begin{proof}
        \begin{after}
        We appeal to an \(L^p\) version of Barbalat's lemma \cite[Fact 3]{teel_asymptotic_1999} \cite{farkas_variations_2016}.
        Conditions i)--iii) on \(\beta(|\dot y|)\) and iv) on \(\beta_1(\dot y)\) suffice to show that \(y(t) \to 0\).
        \end{after}

        Because the system is NHP(1), we have
        \begin{align}
                \int_0^\infty \alpha(|y|) \dif t &\leq \int_0^\infty \beta(|\dot u(t)|)\dif t + V(x(0), u(0))
                \notag \\
                &< \infty.
        \end{align}
        Likewise, the dissipativity with respect to \(\alpha_1(|\dot y|)\) shows that \(\dot y \in L^p\). 

\end{proof}

As a special case, we immediately get asymptotic output stability of \(y\) under steady-state \(u\):
\begin{corollary}
        Suppose that the system \eqref{eq:state_space} is ANHP(1).
        If there exists \(T_0 \geq 0\) such that \(u\) is constant on \([T_0, \infty)\), then \(y(t) \to 0\).
\end{corollary}

\begin{remark}
        Our notion of NHP(1) is related to ``ISS [input-to-state-stability] in \(\dot u\)'' \cite[\S5]{angeli_inputstate_2003}.
        But ISS only implies boundedness, not asymptotic stability.
        Some side information is needed.
        For this one needs ANHP(1).
\end{remark}

\section{Examples: SISO filters}
\label{sec:filters}
We give four variations of the basic first-order linear high pass filter and demonstrate that they are ANHP(1).
Throughout, \(\lambda > 0\) is a rate constant.

\subsection{A linear high-pass filter}
\label{sec:linear_high_pass_example}
In this section, we consider the linear first-order high-pass filter with transfer function \(s/(s + \lambda)\), which can be realized as a circuit where \(u\) is the voltage applied to a capacitor and resistor in series, and \(y\) is the voltage of the resistor.

Even though this system is linear with an evidently high-pass transfer function, we go through the motions in order to illustrate the time-domain technique.
\begin{subequations}
        \label{eq:linear_high_pass_example}
\begin{align}
  \dot x &= \lambda y, \\
  y &= u - x.
\end{align}
\end{subequations}
\begin{proposition}
        The system \eqref{eq:linear_high_pass_example} is NHP(1) with  \(L^2\) supply rate \(s(y, \dot u) = \gamma^2 \dot u^2 - y^2\) for \(\gamma = \frac{1}{\lambda}\).
\end{proposition}
\begin{proof}
        Take the storage function \(V = A y^2\) with \(A\) to be determined.
        Computing time derivatives, we get
        \begin{align}
                \dot V 
                &= -2A \lambda y^2 + 2A y \dot u \\
                \intertext{By Young's inequality, for all \(\epsilon > 0\),}
                &\leq -2A \lambda y^2 + \epsilon A y^2 + \frac{A}{\epsilon} \dot u^2
        \end{align}
        To fulfill dissipativity \(\dot V \leq s(y, \dot u)\):
        \begin{align}
                (-2A \lambda + \epsilon A + 1) y^2 + \del{\frac{A}{\epsilon} - \gamma^2} \dot u^2 &\leq 0
        \end{align}
        Picking \(A = \gamma^2 \epsilon\), we get
        \begin{align}
                -2 \lambda \gamma^2 \epsilon + \gamma^2 \epsilon^2 + 1 &\leq 0
                \label{eq:linear_high_pass_example_dissipativity_inequality}
        \end{align}
        The choice of \(\epsilon = \lambda\) yields the optimal \(\gamma = \frac{1}{\lambda}\).
\end{proof}
\begin{proposition}
        The system \eqref{eq:linear_high_pass_example} is ANHP(1) with supply rate \(\gamma_1^2 \dot u^2 - \dot y^2\) for \(\gamma_1 = 2\).
\end{proposition}
\begin{proof}
        Use the storage function Ansatz \(V_1 = A y^2\) for some \(A > 0\).
        Computing time derivatives, we get
        \begin{align}
                \begin{split}
                        \MoveEqLeft \dot V_1 - \gamma_1 \dot u^2 + \dot y^2\\
                        &= 2 A y (-\lambda y + \dot u) - \gamma_1^2 \dot u^2 + \dot y^2
                \end{split}
                \\
                &\leq (-2A \lambda + \epsilon A + 2 \lambda^2) y^2 + \del{\frac{A}{\epsilon} - \gamma_1^2 + 2} \dot u^2
                \\
                \intertext{Picking \(A = (\gamma_1^2 - 2)\epsilon\), we get}
                &\leq \del{- 2 \epsilon \lambda (\gamma_1^2 - 2) + \epsilon^2 (\gamma_1^2 - 2)  + 2 \lambda^2} y^2 
                \label{eq:linear_high_pass_example_dissipativity_inequality_asymptotic}
        \end{align}
        After choosing \(\epsilon = \lambda\), we have the requirement
        \begin{align}
                - \lambda^2 (\gamma_1^2 - 2)  + 2 \lambda^2 &\leq 0
        \end{align}
        which is satisfied by all \(\gamma_1 \geq 2\).
\end{proof}

\subsection{A nonlinear high-pass filter}
\label{sec:nonlinear_high_pass_example}
Let us reimagine the RC circuit of the last section with a nonlinear resistor having the current-voltage characteristic \(i \propto \sinh v\), which can be realized as two cross-polarized ideal diodes in parallel.
This circuit can arise when designing an alternating current light-emitting-diode lamp with a capacitive ballast.
(See \cite{nerone_led_2003} and patents cited therein.)
In such a scenario, we would like to verify formally that the average power \(iv\), expressed as light plus heat, does not exceed safe dissipation limits.
\begin{subequations}
        \label{eq:nonlinear_high_pass_example}
        \begin{align}
                \dot x &= \lambda \sinh y, \\
                y &= u - x
        \end{align}
\end{subequations}

\begin{proposition}
        The system \eqref{eq:nonlinear_high_pass_example} is NHP(1) with supply rate \(\gamma \dot u \arsinh \dot u - y \sinh y\) for gain \(\gamma = \frac{2}{\lambda}\).
\end{proposition}
\begin{skippedproof}
        The key to this inequality is a nonlinear bound on the cross term \(y\dot u\):
        \begin{align}
                |y\dot u|
                &\leq \frac{\lambda}{2}  y \sinh y + \dot u \arsinh (\dot u/\lambda).
                \label{eq:nonlinear_high_pass_example_dissipativity_inequality}
        \end{align}
        This follows from Fenchel's inequality.

        We check dissipation using the storage function Ansatz \(V = \frac{\gamma}{2} y^2\).
        Computing time derivatives, we get
        \begin{align}
                \begin{split}
                        \MoveEqLeft \dot V -\gamma \dot u \arsinh (\dot u/\lambda )+ y \sinh y
                        \\
                        &= \gamma y (-\lambda \sinh y + \dot u) -\gamma \dot u \arsinh \dot u + y \sinh y
                \end{split}
                \\
                \intertext{Using \eqref{eq:nonlinear_high_pass_example_dissipativity_inequality},}
                \begin{split}
                        &\leq
                        -\gamma \lambda y \sinh y + \frac{\gamma}{2} \lambda y \sinh y \\
                        &\quad + \gamma \dot u \arsinh (\dot u /\lambda) -\gamma \dot u \arsinh \dot u + y \sinh y
                \end{split}
                \\
                &= \del{-\frac{\gamma}{2} \lambda + 1} y \sinh y
        \end{align}
        which is nonpositive as long as \(\gamma \geq \frac{2}{\lambda}\).
\end{skippedproof}

\begin{remark}
        \begin{after}
                In the preceding proof,  the storage function corresponds to the electrostatic energy stored in the capacitor.
                The supply rate \(\gamma \dot u \arsinh \dot u - y \sinh y\) implies \(L^2\) stability as \(\dot u \arsinh \dot u \leq \dot u^2\) and \(y \sinh y \geq y^2\).

        \end{after}

\end{remark}

\begin{proposition}
        The system \eqref{eq:nonlinear_high_pass_example} is ANHP(1) using the supply rate \(\gamma_1^2 \dot u^2 - \dot y^2\) with \(\gamma_1 = 2\).
\end{proposition}
\begin{skippedproof}
        Use the storage function Ansatz \(V = 2A\cosh y\).
        \begin{align}
                \begin{split}
                        \MoveEqLeft
                        \dot V - \gamma_1^2 \dot u^2 + \dot y^2
                        \\
                        &\leq 2A \sinh y (-\lambda \sinh y + \dot u) - \gamma_1^2 \dot u^2 + \dot y^2
                \end{split}
                \\
                \begin{split}
                        &\leq
                        \del{-2 \lambda A + \epsilon A + 2 \lambda^2} \sinh^2 y
                        \\
                        &\quad + \del{\epsilon^{-1} A - \gamma_1^2 + 2} \dot u^2
                \end{split}
        \end{align}
        Now we can recycle the inequality \eqref{eq:linear_high_pass_example_dissipativity_inequality_asymptotic}.
\end{skippedproof}

\subsection{A slow nonlinear high-pass filter}
\label{sec:slow_nonlinear_high_pass_example}
As a variation on the previous themes, we study a ``slow manifold'' input-output system:
\begin{subequations}
\label{eq:slow_nonlinear_high_pass_example}
        \begin{align}
                \dot x &= \lambda y^3 \\
                y &= u - x. 
        \end{align}
\end{subequations}

\begin{proposition}
        The system \eqref{eq:slow_nonlinear_high_pass_example} is NHP(1) with the supply rate \(\gamma^{4/3} \dot u^{4/3} - y^4\) for \(\gamma = \frac{1}{\lambda}\).
\end{proposition}
\begin{skippedproof}
        We use the storage function \(V = 2A y^2\).
        \begin{align}
                \begin{split}
                        \MoveEqLeft 
                        \dot V - \gamma^4 \dot u^4 + y^4
                        \\
                        &\leq - 4A \lambda y^4 + 4A y \dot u - \gamma^{4/3} \dot u^{4/3} + y^4
                \end{split}
                \intertext{Use the general form of Young's inequality where \(ab \leq \frac{\epsilon^p a^p}{p} + \frac{b^q}{\epsilon^q q}\) for \(\frac{1}{p} + \frac{1}{q} = 1\). For any \(\epsilon > 0\),}
                &\leq \del{-4A \lambda + A \epsilon^4 + 1} y^4 + (3A \epsilon^{-4/3} - \gamma^{4/3}) \dot u^{4/3}
                \intertext{Setting \(A = \gamma^{4/3} \epsilon^{4/3} / 3\), we get}
                &\leq \del{-\frac{4}{3} \lambda \gamma^{4/3} \epsilon^{4/3} + \frac{1}{3} \gamma^{4/3} \epsilon^{\frac{16}{3}} + 1} y^4 
                \intertext{The choice of \(\epsilon^{4/3} = \lambda^{1/3}\) yields}
                &\leq \del{-\gamma^{4/3} \lambda^{4/3} + 1} y^4
        \end{align}
        from which we get \(\gamma = \frac{1}{\lambda}\).        
\end{skippedproof}

\begin{proposition}
        The system \eqref{eq:slow_nonlinear_high_pass_example} is ANHP(1) with the supply rate \(\gamma_1^2 \dot u^2 - \dot y^6\) for some \(\gamma_1 > 0\).
\end{proposition}
\begin{skippedproof}
        We compute \(\dot y^2 \leq 2\dot u^2 + 2\lambda^2 y^6\) and omit computations using the storage function \(V_1 \propto y^4\).
\end{skippedproof}
\begin{remark}
        It may be aesthetically unappealing that the storage rates in this case do not feature the same exponents for ``input'' and ``output.''
        However, this asymmetry can be justified on the basis of dimensional (or scaling) analysis: unlike in the linear example \eqref{eq:linear_high_pass_example}, \(\lambda\) does not have units of pure time.
\end{remark}

\subsection{A sector-bounded nonlinear resistor}
Given \(\alpha, \beta \geq 0\), a function \(\phi\) is said to be in the sector \([\alpha, \beta]\) if the graph of \(p = \phi(q)\) is sandwiched between the lines \(p = \alpha q\) and \(p = \beta q\).
This popular class of almost-linear functions was originally introduced to study the stability of nonlinear feedback interconnections \cite[Theorem 10.5]{astrom_feedback_2008}, and has recently become useful for certifying neural networks with activation functions such as ReLU and tanh \cite{junnarkar_synthesizing_2024}.

The following system represents the analog of an RC high-pass filter implemented using a sector-bounded nonlinear resistor:
\begin{subequations}
        \label{eq:sector_bounded_nonlinear_high_pass_example}
        \begin{align}
                \dot x &= \lambda \phi(y) \\
                y &= u - x.
        \end{align}
\end{subequations}
\begin{proposition}
        The system \eqref{eq:sector_bounded_nonlinear_high_pass_example} is ANHP(1)  if \(0 < \alpha, \beta <\infty\). 
\end{proposition}
\begin{proof}
        We use the characterization of a sector-bounded function via a quadratic inequality:
        \begin{subequations}
                \begin{align}
                        (p - \alpha q)(p - \beta q) \leq 0
                \end{align}
                from which we obtain the estimates:
                \begin{align}
                        pq &\geq \frac{p^2 + \alpha\beta q^2}{\alpha + \beta}
                        \label{eq:sector_bounded_nonlinear_high_pass_characterization_power}
                        \\
                        p^2 &\leq \alpha\beta q^2 - (\alpha + \beta) pq \leq \alpha\beta q^2
                        \label{eq:sector_bounded_nonlinear_high_pass_characterization_quadratic}
                \end{align}
        \end{subequations}
        To verify NHP(1), we use the storage function \(V  = A y^2\) and supply rate \(\gamma^2 \dot u^2 + y^2\):
        \begin{align}
                \dot V - \gamma^2 \dot u^2 + y^2
                &= A y \del{-\lambda \phi(y) + \dot u} - \gamma^2 \dot u^2 + y^2
        \end{align}
        Next we apply \eqref{eq:sector_bounded_nonlinear_high_pass_characterization_power} to get \(-\lambda y A \phi(y) \leq \frac{-\lambda A \alpha \beta}{\alpha + \beta} y^2\) and omit the remainder of the calculation.

        To verify ANHP(1), we use the supply rate \(\gamma^2 \dot u^2 + \dot y^2\).
        Differentiating and applying \eqref{eq:sector_bounded_nonlinear_high_pass_characterization_quadratic}, we have
        \begin{align}
                \dot y^2
                &\leq 2 \dot u^2 + 2 \lambda^2 \alpha\beta y^2
        \end{align}
        and omit the rest of the steps.
\end{proof}

\section{Application to robust control}
\label{sec:robust_control}
We examine the problem of robustly stabilizing the origin of the first-order integrator system
\begin{align}
        \dot y = u + d
        \label{eq:1st_order_integrator_system}
\end{align}
where \(y\) is the regulated output, \(u\) is the control input, and \(d\) is a disturbance.
Seeking steady-state rejection of the \(d\), we design 
a control law with nonlinear proportional and integral feedback:
\begin{align}
        u(t) = -G(y(t)) - \int^t f(y(s)) \dif s,
        \label{eq:control_law}
\end{align}
where \(G\) and \(F\) are differentiable functions, \(f = F'\), \(g = G'\), and the following conditions hold:
\begin{align*}
        G(0) = F(0) = 0, \ 
        g(y) \geq g_0 > 0, \ \text{and}\  y f(y) > 0.
\end{align*}

Henceworth we combine \eqref{eq:1st_order_integrator_system} and \eqref{eq:control_law} and differentiate to get closed loop dynamics
\begin{align}
\label{eq:robust_closed_loop_example_ii}
        \ddot y
        + g(y)  \dot y + f(y)
        &= \dot d.
\end{align}
(To recover linear PI control, one may choose linear \(G\) and quadratic \(F\).)

Specify a gain function \(\psi\) by
\begin{align}
        \psi(s) &= \sup_{r\in\mathbb{R}} \cbr{G(r) \del{2 s - f(r)}}.
        \label{eq:psi}
\end{align}
\begin{proposition}
        The storage function \(V = \del{\dot y + G(y)}^2 + 4 F(y) + \dot y^2\)
        satisfies \(\dot V \leq -f(y) G(y) - g_0 \dot y^2 + \psi(\dot d) + \frac{4\dot d^2}{g_0}\).
\end{proposition}
\begin{proof}
        This construction initially follows \cite[Remark]{blinov_estimating_2001}.
        By direct computation.
        \begin{align}
                \dot V &= - 2f(y) G(y) - 2g(y) \dot y^2 + 4 \dot y \dot d + 2 G(y) \dot d
                \\
                &\leq - 2f(y) G(y) - 2g_0 \dot y^2 + 4 \dot y \dot d + 2 G(y) \dot d
                \\
                \begin{split}
                        &= -f(y) G(y) - g_0 \dot y^2 
                        \\
                        &\quad + \cbr{- f(y) G(y) - g_0 \dot y^2 + 4 \dot y \dot d + 2 G(y) \dot d}
                \end{split}
                \\
                \begin{split}
                        &\leq
                        -f(y) G(y) - g_0 \dot y^2 
                        \\
                        &\quad + \sup_{r, s}\cbr{- f(r) G(r) - g_0 s^2 + 4 s \dot d + 2 G(r) \dot d}
                \end{split}
                \\
                &\leq -f(y) G(y) - g_0 \dot y^2 + \psi(\dot d) + \frac{4\dot d^2}{g_0}.
        \end{align}
\end{proof}

\begin{proposition}
        \label{prop:robust_closed_loop_example_ii}
        The system \eqref{eq:robust_closed_loop_example_ii} is NHP(1) and, moreover, ANHP(1) from \(\dot d\) to \(y\) with supply rate \(\psi(\dot d) + \frac{4\dot d^2}{g_0} - f(y) G(y) - g_0\dot y^2\).
\end{proposition}

\subsection{Example: nonsmooth PI control}
Let us consider the system \eqref{eq:1st_order_integrator_system} with the control law
\eqref{eq:control_law} realized as PI control with a discontinuous integral term:
\begin{subequations}
        \label{eq:nonsmooth_pi_control}
        \begin{align}
                G(y) &= k_p y,
                \\
                f(y) &= k_i y + k_s \sgn{y},
        \end{align}
\end{subequations}
for \(k_p, k_i, k_s > 0\).

To compute \(\psi\) pursuant to \eqref{eq:psi}, note that
\begin{align}
        \psi(s) &= \sup_{r\in\mathbb{R}} \cbr{k_p r \del{2 s - k_i r - k_s \sgn{r}}}
        \\
        &= \frac{\max(|2s| - k_s, 0)^2}{4 k_i/k_p}.
\end{align}
Now we invoke Prop.~\ref{prop:robust_closed_loop_example_ii} to claim that, modulo initial conditions,
\begin{multline}
        \int_0^T \del{k_p k_i y^2 + k_p k_s |y| + k_p \dot y^2} \dif t
        \\
        \leq \int_0^T \del{
                \frac{\max(|2\dot d| - k_s, 0)^2}{4 k_i/k_p} + \frac{4 \dot d^2}{k_p}
        }\dif t
        + \ldots
\end{multline}
This gain inequality implies that if the integral on the RHS converges as \(T \to \infty\), then \(y \to \infty\).

\section{Discussion: higher-order NHP behavior and compositionality}
\label{sec:higher_order_nhp}
Two dissipative systems can be connected in series.\footnote{For applications to large-scale feedback systems, see \cite{arcak_networks_2016}.}
Suppose that system \(y_1 = \mathcal{H}_1 u\) is  dissipative with supply rate \(s_1(y_1, u)\) and system \(y_2 = \mathcal{H}_2 y_1\) is dissipative with supply rate \(s_2(y_2, y_1)\).
Then the composite \(y_2 = \mathcal{H}_2 \mathcal{H_1} u\) is dissipative with supply rate \(s_1(y_1, u) + s_2(y_2, y_1)\).

When it comes to high-pass behavior, the natural concern regarding compositionality is whether the series interconnection of two first-order high pass filters is a second-order high-pass filter.
In the following discussion we fix a class-\(\mathcal{K}_\infty\) function \(\alpha(x)\) that grows faster than \(x\).
We use the notation \(\left\|\mathcal{H}\right\|_{\alpha} \leq \gamma\) to mean that \(y = \mathcal{H} u\) is dissipative with supply rate \(\gamma \alpha(u) - \alpha(y)\).
By adding supply rates, we have composition law for system norms, \(\left\|\mathcal{H}_1 \mathcal{H}_2\right\|_\alpha \leq \left\|\mathcal{H}_1 \right\|_\alpha \left\|\mathcal{H}_2\right\|_\alpha\).

Let \(p = \od{}{t}\) be the differentiation operator.
The definition of NHP(1) can be restated as: there exists \(\gamma_1\) such that
\begin{align*}
        \left\|\mathcal{H} p^{-1}\right\|_\alpha &\leq \gamma_1 < \infty.
\end{align*}
Likewise, ANHP(1) can be restated as \(\left\|p \mathcal{H} p^{-1}\right\|_\alpha < \infty\).

\begin{theorem}
        Suppose that systems \(\mathcal{G}\) and \(\mathcal{H}\) are ANHP(1) with respect to \(\alpha\).
        \begin{align*}
                \left\|\mathcal{H} p^{-1}\right\|_\alpha, \left\|p\mathcal{H} p^{-1}\right\|_\alpha &< \infty \\
                \left\|\mathcal{G} p^{-1}\right\|_\alpha, \left\|p\mathcal{G} p^{-1}\right\|_\alpha &< \infty \\
        \end{align*}
        Then the composite \(\mathcal{G} \mathcal{H}\) is also ANHP(1).
\end{theorem}
\begin{proof}
        We directly compute the NHP(1) norm:
        \begin{align}
                \left\|GH p^{-1}\right\|_\alpha
                &= \left\|G p^{-1} p H p^{-1}\right\|_\alpha
                \\
                &\leq \left\|G p^{-1}\right\|_\alpha \left\| p H p^{-1}\right\|_\alpha
                \\
                &< \infty
        \end{align}
        as well as the ANHP(1) norm:
        \begin{align}
                \left\|p GH p^{-1}\right\|_\alpha
                &= \left\|pG p^{-1} p H p^{-1}\right\|_\alpha
                \\
                &\leq \left\|pG p^{-1}\right\|_\alpha \left\| p H p^{-1}\right\|_\alpha
                \\
                &< \infty
        \end{align}
\end{proof}

It is tempting to define: for some \(k\geq 0\), a system \(\mathcal{H}\) is \(\text{NHP}(k)\) if there exists \(\gamma_k\) such that
\begin{align*}
        \left\|\mathcal{H} p^{-k}\right\|_\alpha &\leq \gamma_k < \infty.
        \intertext{Sadly, this tentative definition does not work. Suppose that \(\mathcal{H}\) and \(\mathcal{G}\) are both NHP(1). We would like to conclude that}
        \left\|\mathcal H \mathcal G p^{-2}\right\|_\alpha
        &\overset{?}{\leq}
        \left\|\mathcal H p^{-1} \mathcal G p^{-1}\right\|_\alpha
        \\
        &\leq \left\|\mathcal H p^{-1}\right\|_\alpha \left\| \mathcal G p^{-1}\right\|_\alpha
        \intertext{Unfortunately, the first inequality does not hold in general because unlike linear systems \cite[Chapter 6, exercise 6.1]{astrom_feedback_2008}, nonlinear systems do not necessarily commute with \(p\). A possible fix is}
        \left\|\mathcal H \mathcal G p^{-2}\right\|_\alpha
        &\leq \left\|\mathcal H p^{-1} p \mathcal G p^{-2}\right\|_\alpha
        \\
        &\leq \left\|\mathcal H p^{-1} \right\|_\alpha \left\|p \mathcal G p^{-2}\right\|_\alpha,
\end{align*}
which motivates the following definition:
\begin{definition}
        A nonlinear operator \(\mathcal{G}\) is \(\text{NHP}(k)\) with smoothness \(\ell\) if for all \(0 \leq j \leq \ell - 1\),
        there exists constants \(\gamma^{j}_\ell\) such that
        \begin{align*}
                \left\|p^{j}\mathcal{H} p^{j-k}\right\|_\alpha \leq \gamma_j < \infty.
        \end{align*}
\end{definition}
This definition is unwieldy in practice, and while we it may be amenable to automatic formal verification methods such as sum-of-squares programming, we are not presently aware of any strictly nonlinear operators that can be hand-checked to satisfy this definition for \(k > 1\).

\section{Conclusion}
Based on the observation that a linear high-pass system remains stable when composed with an integrator, we have proposed a definition of nonlinear high-pass behavior and
exemplified how this point of view can be fruitful for new, quantitative analytical methods for nonlinear circuits and feedback systems.
These methods build on Lyapunov-style analysis using Barbalat's lemma and  issue stronger and more detailed guarantees on system safety and performance.
We believe that the NHP formalism can be a worthwhile application for formal verification methods based on semidefinite programming, such as sum-of-squares programming \cite{ebenbauer_analysis_2006} and neural  networks \cite{junnarkar_synthesizing_2024}.


\addtolength{\textheight}{0cm}   







\bibliography{IEEEabrv,export}
\bibliographystyle{IEEEtran}

\end{document}